\documentclass[journal]{IEEEtran}
\IEEEoverridecommandlockouts                          
\usepackage[english]{babel}
\usepackage{amsmath, amssymb, amsfonts, amsthm, bm}
\usepackage{mathrsfs}
\usepackage[mathcal]{euscript}
\usepackage{enumerate}
\usepackage{mathbbol}
\usepackage{float}
\usepackage{color,graphicx,epstopdf}
\usepackage{dsfont}
\usepackage{hyperref}
\usepackage{caption}
\usepackage{subcaption}
\usepackage{cite} 
\usepackage{algorithm} 
\usepackage{algorithmic} 
\usepackage{tikz-cd}
\usepackage{mathtools}
\usepackage{lipsum}

\newcommand*\mcapinn[2]{\vcenter{\hbox{$\mathsurround=0pt
			\ifx\displaystyle#1\textstyle\else#1\fi\bigcap$}}}
\newcommand*\mcup{\mathbin{\mathpalette\mcupinn\relax}}
\newcommand*\mcupinn[2]{\vcenter{\hbox{$\mathsurround=0pt
			\ifx\displaystyle#1\textstyle\else#1\fi\bigcup$}}}

\newtheorem{theorem}{Theorem}
\newtheorem{definition}{Definition}

\newtheorem{lemma}{Lemma}

\newtheorem{proposition}{Proposition}

\title{Social Shaping for Transactive Energy Systems}

\author{Zeinab Salehi, Yijun Chen, Ian R. Petersen, Elizabeth L. Ratnam, Guodong Shi}

\begin{document}
\maketitle

\begin{abstract}
This paper considers the problem of shaping agent utility functions in a transactive energy system to ensure the optimal energy price at a competitive equilibrium is always socially acceptable, that is, below a prescribed threshold. Agents in a distributed energy system aim to maximize their individual payoffs, as a combination of the utility of energy consumption and the income/expenditure from energy exchange. The utility function of each agent is parameterized  by individual preference vectors, with the overall system operating at competitive equilibriums. We show the social shaping problem of the proposed transactive energy system is conceptually captured by a set decision problem. The set of agent preferences that guarantees a socially acceptable price  is characterized by an implicit algebraic equation for strictly concave and continuously differentiable utility functions. We also present two analytical solutions where  tight ranges for the coefficients of   linear-quadratic utilities and piece-wise linear utilities are established under which optimal pricing is proven to be always socially acceptable. 
\end{abstract}


\section{Introduction}
 
Recent dramatic increases in rooftop solar, backed by energy storage technologies including electrical vehicles and home batteries, are obfuscating the traditional boundary between energy producer and consumer \cite{Lin2019Comparative}. At the same time, transactive energy systems are being designed to coordinate supply and demand for various distributed energy sources in an electrical power network \cite{Wang2018Transactive}, where power lines (or cables) enable electricity exchanges, and communication channels enable cyber information exchanges. With transactive energy systems supporting the rise of the prosumer (i.e., an electricity producer and consumer), new ways are emerging to enable a fair \cite{Brockway2021Inequitable} and sustainable renewable energy transition \cite{Anastopoulou2021Efficient}.

In the recent literature, transactive energy systems have been designed for a wide-range of power system applications. For example, transactive energy systems are supporting microgrid operations \cite{Akter2020optimal}, virtual power plant operations \cite{BEHBOODI2018Transactive}, and the operation of bulk power systems with a proliferation of renewable and distributed energy resources \cite{Werner2021Pricing}. Such transactive energy systems are primarily focused on market-based approaches for balancing electricity supply and demand, supporting robust frequency regulation throughout the bulk grid \cite{King2021Solving}. Typically, agents operating in transactive energy subsystem (e.g., virtual power station) are located across a power network and market mechanisms are in place to support individual preferences, enabling agents to compete and collaborate with each other \cite{Pierluigi2019Survey}. The aim is a dramatic improvement in power systems operational efficiency, scalability, and resilience. The market mechanisms used in transactive energy networks to enable efficient and valuable energy transactions are typically drawn from classical theories in economics and game theory \cite{Li2020}. 

In standard welfare economics theory \cite{Debreu1952} it is suggested that resources pricing can be designed to balance supply and demand in a market. In a multi-agent system with distributed resource allocations, agents have the autonomy to decide their local resource consumption and exchange to optimize their individual payoffs as a combination of local utility and income or expenditure. In a competitive equilibrium, resource pricing is achieved when all agents maximize their individual payoffs subject to a network-level supply-demand balance constraint, which in turn maximizes the overall system-level payoff  \cite{Acemoglu}. More specifically, to balance network supply and demand, resource pricing corresponds to the optimal dual variables associated with supply-demand balance constraint, where the price also maximizes the system-level payoff. The prospect  of operating a transactive energy system as a market via optimal pricing under a competitive equilibrium has been widely studied in \cite{Chen10, Li2011, Papadaskalopoulos13, Hansen15, Li2015, Jadhav2018, Muthirayan2019}. 

Early and recent studies in both the economics and engineering literature, are primarily focused on controlling resources price volatility \cite{Black1973, Heston1993, Hagan2002, Boll2002}, rather than the resources price itself. Consequently, in practice, the optimally computed price to balance supply and demand in a market is potentially not socially acceptable. For example, in February 2021, the wholesale electricity price in Texas was considered by prosumers to be unacceptably high after widespread power outages. Prosumers that subscribed to the wholesale price reported receiving sky-high electricity bills, which greatly exceeded societal expectations for payment \cite{Blumsack}.   

Several authors have referred to price volatility as a rapid change in the pricing process. The Black Scholes formula \cite{Black1973} and Heston's extension \cite{Heston1993} are the most representative models of stochastic price volatility, while others include SABR \cite{Hagan2002} and GARCH \cite{Boll2002}. The authors in \cite{Kizi2010} argued that previous models did not penalize price volatility in the system-level objective, so they proposed to modify the system-level objective to account for price volatility, whereby constructing a dynamic game-theoretic framework for power markets. The authors in \cite{Tsitsi2015} proposed a different dynamic game-theoretic model for electricity markets, by incorporating a pricing mechanism with the potential to reduce peak load events and the cost of providing ancillary services. In \cite{Wei2014}, the dual version of the system-level welfare optimization problem was considered, where an explicit penalty term on the $L_2$ norm of price volatility was introduced in the system-level objective, allowing for trade-offs between price volatility and social welfare considerations.

In this paper, we focus on resource pricing rather than price volatility, to support the design of socially acceptable electricity markets. We propose a social shaping problem for a competitive equilibrium in a transactive energy systems, aiming to bound the energy price below a socially acceptable threshold. We focus on parameterized utility functions, where the parameters are abstracted from the preferences of
agents. We prescribe a range for the parameters in the utility functions, to ensure resources pricing under a competitive equilibrium is socially acceptable for all agents without creating a mismatch in supply and demand. The idea of introducing parameterized utility functions is informed by the concept of smart thermostat agents in the AEP Ohio gridSMART demonstration project \cite{Widergren2014D, Widergren2014R}. For transactive energy systems organized as multi-agent networks operating at competitive equilibriums, we establish the following results. 
\begin{itemize}
    \item We show the essence of the social shaping problem is a set decision problem, where for  strictly concave and continuously differentiable utility functions, the set decision is characterized by an implicit representation of the optimal price as a function of agent preferences. 
    \item For two representative classes of utility functions, nam-ely linear-quadratic functions and piece-wise linear functions, the exact set of parameters that guarantee socially acceptable pricing is established, respectively. 
\end{itemize}

In our prior work \cite{CDC}, we introduced the concept of social shaping of agent utility functions, which was followed by an investigation into social shaping with linear quadratic utility functions \cite{ANZCC}. In the current manuscript, the results on conceptual solvability and solutions for piece-wise linear functions are new, which are supported by the presentation of a series of new and large-scale numerical examples. 

This paper is organized as follows. In Section~\ref{rom2}, we introduce two multi-agent transactive energy systems, which are implemented in either a centralized or distributed transactive energy network. In Section~\ref{rom3}, we motivate the problem of social shaping, to enable the design of socially acceptable electricity pricing. In Section~\ref{rom4} and Section~\ref{rom5}, we present conceptual and analytical solutions to the social shaping problem, respectively. Section~\ref{rom6} presents concluding remarks.

\section{Multiagent Transactive Energy Systems}\label{rom2}
In this section, we introduce two multi-agent transactive energy systems models, and we recall some fundamental definitions and results related to such systems.  

\subsection{Transactive Energy Systems as Multi-agent Networks}

We present two simple yet representative setting where  $n$ agents indexed in the set $\mathrm{V}=\{ 1, ..., n \}$, each with local energy supply and demand, form a transactive energy system. The multi-agent network is designed to support microgrid operations, or would support distribution grid operators with managing ubiquitous behind-the meter renewable energy resources. For simplicity, we assume a lossless electrical network, with extensions to include both real and reactive network and load losses possible. When there is no external energy resource (e.g., a microgrid operating in isolation from the wider power grid), the $n$ agents seek to form an energy market where the overall energy supply and demand are balanced. Such a market must incorporate the diverse  interests of individual agents, while ensuring market efficiency.

\vspace{2mm}

\noindent{\it Multiagent Transactive Energy Systems (MTES)}.
Each agent $i$ produces electricity with a local energy resource, with $a_i \in \mathbb{R}^{\geq 0}$ representing \emph{local power production} (in kW). The overall \emph{network generation production} is represented  by $C := \sum_{i=1}^n a_i$, where $C>0$ (in kW). Each agent $i$ makes a decision on its \emph{energy consumption load}, denoted by  $x_i \in \mathbb{R}^{\geq 0}$ (in kW). Upon consuming the load $x_i$, agent $i$ receives a payment (or bill) attributed to its demand preference $f_i(x_i) = h(x_i; \theta_i)$, where $\theta_i$ is the \emph{personalized parameter} for the load preferences of agent $i$. Importantly, any shortfall or surplus of energy for each agent $i$, represented by $a_i-x_i$, must be accommodated by the transactive market. We denote the \emph{price per unit energy} by $\lambda\in \mathbb{R}$ (in \$$\backslash$kWh), and the income (or cost) of a transaction for agent $i$ is thereby $\lambda(a_i-x_i)$ when $a_i>x_i$ (or $a_i<x_i$).

\vspace{2mm}

\noindent{\it Multiagent Transactive Energy Systems with Strategic Trading (MTES-ST)}. We extend the previously defined MTES, by way of supporting \emph{strategic trading decisions} for each agent $i$ denoted by $e_i \in \mathbb{R}$. That is, the income (or cost) of the transactions for agent $i$ is $\lambda e_i$, where $e_i>0$ (or $e_i<0$). Importantly, there is an inherent constraint on strategic trading decisions for each agent, represented by $e_i\leq a_i-x_i$. Specifically, when $a_i>x_i$, there is a physical constraint on $e_i$ as the amount of energy sold by agent $i$ cannot exceed $a_i-x_i$. Furthermore, when $a_i<x_i$, there is a physical network constraint on $e_i$ for agent $i$ seeking to purchase $x_i-a_i$ amount of energy from the market.   

\subsection{System-level Equilibriums}
We draw on welfare theory from economics, considering an effective market price in the context of rational agent decisions. Specifically, for both MTES and MTES-ST, we consider the concepts of competitive equilibriums and social welfare equilibriums to precisely quantify price effectiveness and agent rationality.

Let $\mathbf{a}=(a_1, ..., a_n)^\top \in \mathbb{R}^n$ be a vector representing the \emph{local production profile}, or otherwise the electricity supply available from $n$ agents. Let $\mathbf{x}=(x_1, ..., x_n)^\top \in \mathbb{R}^n$ be a vector representing the \emph{local consumption profile},  or otherwise the electricity demand for $n$ agents.

\begin{definition}
	For the MTES, a competitive equilibrium is the pair of (1) price denoted by $\lambda^\ast\in \mathbb{R} $, and (2) local consumption profile denoted by $\mathbf{x}^\ast \in \mathbb{R}^n$, under which the following two conditions are satisfied.  
	\begin{itemize}
		
		\item[(i)] The local consumption profile $\mathbf{x}^\ast$ maximizes each individual agent payoff; i.e., each $x_i^\ast$ is a solution to the following optimization problem		\begin{equation}\label{opt_LD_1}
			\begin{aligned}
				\max_{{x}_i} \quad &  h(x_i; \theta_i) + \lambda^\ast (a_i -x_i) \\
			{\rm s.t.} \quad & x_i\in \mathbb{R}^{\geq0}.
			\end{aligned}
		\end{equation}
		
		\item[(ii)] The local consumption profile $\mathbf x^\ast$ balances the total energy consumption and supply across the network; that is,
		\begin{equation}\label{load_demand_supply_constraints}
			\sum_{i=1}^n x_i^\ast =C.
		\end{equation}
	\end{itemize}
\end{definition}
 
\begin{definition}
	Let $\mathbf{e}=(e_1, ..., e_n)^\top \in \mathbb{R}^n$ denote a \emph{strategic decision profile}. A competitive equilibrium for  MTES-ST is a triplet of ($\lambda^\ast, \mathbf{x}^\ast, \mathbf{e}^\ast$), under which the following conditions are satisfied.
	\begin{itemize}
		\item[(i)] The pair $(\mathbf x^\ast, \mathbf e^\ast)$ maximizes the individual payoff of each agent, i.e., $(x_i^\ast, e_i^\ast)$ is a solution to 
		\begin{equation}\label{opt_LTD_1}
			\begin{aligned}
				\max_{{x}_i, e_i} \quad &  h(x_i; \theta_i)+\lambda^\ast e_i \\
		{\rm s.t.}  \quad  & x_i+e_i\leq a_i, 
		\\ & x_i\in \mathbb{R}^{\geq0},\ \  e_i\in \mathbb{R}.
			\end{aligned}
		\end{equation}
		\item[(ii)] The strategic decision profile $\mathbf e^\ast$ balances the total resource supply and demand across the network, i.e.,
		\begin{equation}\label{trading_demand_supply_constraints}
			\sum_{i=1}^n e_i^\ast =0.
		\end{equation} 
	\end{itemize}
\end{definition}

The aforementioned competitive equilibriums support the establishment of an effective market with supply and demand balanced. Rational agents decisions are support by way of maximizing their individual payoffs.  
From classical welfare economic theory,  competitive equilibriums guarantee Pareto optimality in the sense that no agent can change a decision without reducing the payoff of other agents  \cite{Acemoglu,Arrow1954,Debreu1952}. Next, we consider optimality at a system-level, in the context of a social welfare problem (i.e., in the absence of a market).

\begin{definition}
(i)	For the  MTES, a social welfare equilibrium $\mathbf{x}^\ast$ is achieved by way of solving the following maximization problem
	\begin{equation}\label{opt_social_LD_1}
		\begin{aligned}
			\max_{\mathbf{x}} \quad &  \sum_{i=1}^{n} h(x_i; \theta_i) \\
			{\rm s.t.} \quad &  \sum_{i=1}^n x_i =C, \\
			&   x_i\in \mathbb{R}^{\geq0}, \, \, i \in V.
		\end{aligned}
	\end{equation}
	
(ii) For the MTES-ST, a social welfare equilibrium $(\mathbf{x}^\ast, \mathbf{e}^\ast)$ is achieved by solving the following maximization problem
	\begin{equation}\label{opt_social_LTD_1}
		\begin{aligned}
			\max_{\mathbf{x}, \mathbf{e}} \quad &  \sum_{i=1}^{n} h(x_i; \theta_i) \\
			{\rm s.t.} \quad &  \sum_{i=1}^n e_i =0, \\
			 \quad & x_i+e_i\leq a_i, \\
			 & x_i\in \mathbb{R}^{\geq0},  \, e_i\in \mathbb{R}, \, \, i \in V. 
		\end{aligned}
	\end{equation}	
\end{definition}

The social welfare equilibrium is consistent with the rationale of a utility-based system planner that designs and enforces energy consumption decisions across all agents. Such a system planner is not concerned by individual payoffs of each specific agent, but rather, the planner aims to maximize energy allocations across the network, even at the expense of some agents receiving suboptimal allocations. In what follows, we investigate conditions where both the competitive and social welfare equilibriums are equivalent. 

\begin{proposition}[as in \cite{CDC}]\label{prop_equilibrium_equivalence}
Suppose each $h(\cdot; \theta_i)$ is a concave function over the domain $\mathbb{R}^{\geq 0}$.	Then for both MTES and MTES-ST, the competitive equilibriums and the social welfare equilibriums are equivalent. 
\end{proposition}

\begin{figure*}
  \centering
  \begin{minipage}[b]{0.4\textwidth}
    \includegraphics[width=\textwidth]{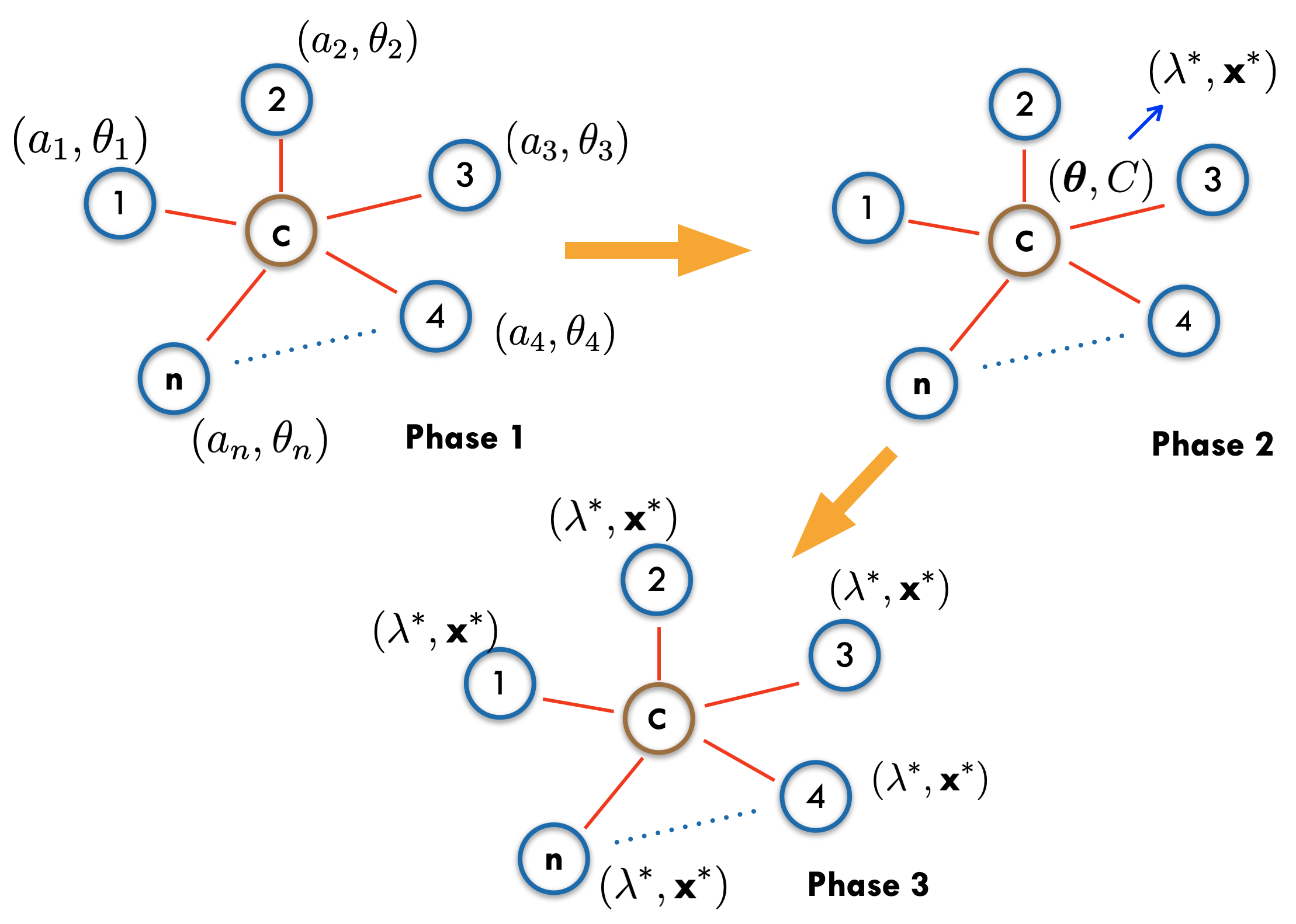}
  \end{minipage}
  \medskip \quad \quad 
  \begin{minipage}[b]{0.4\textwidth}
    \includegraphics[width=\textwidth]{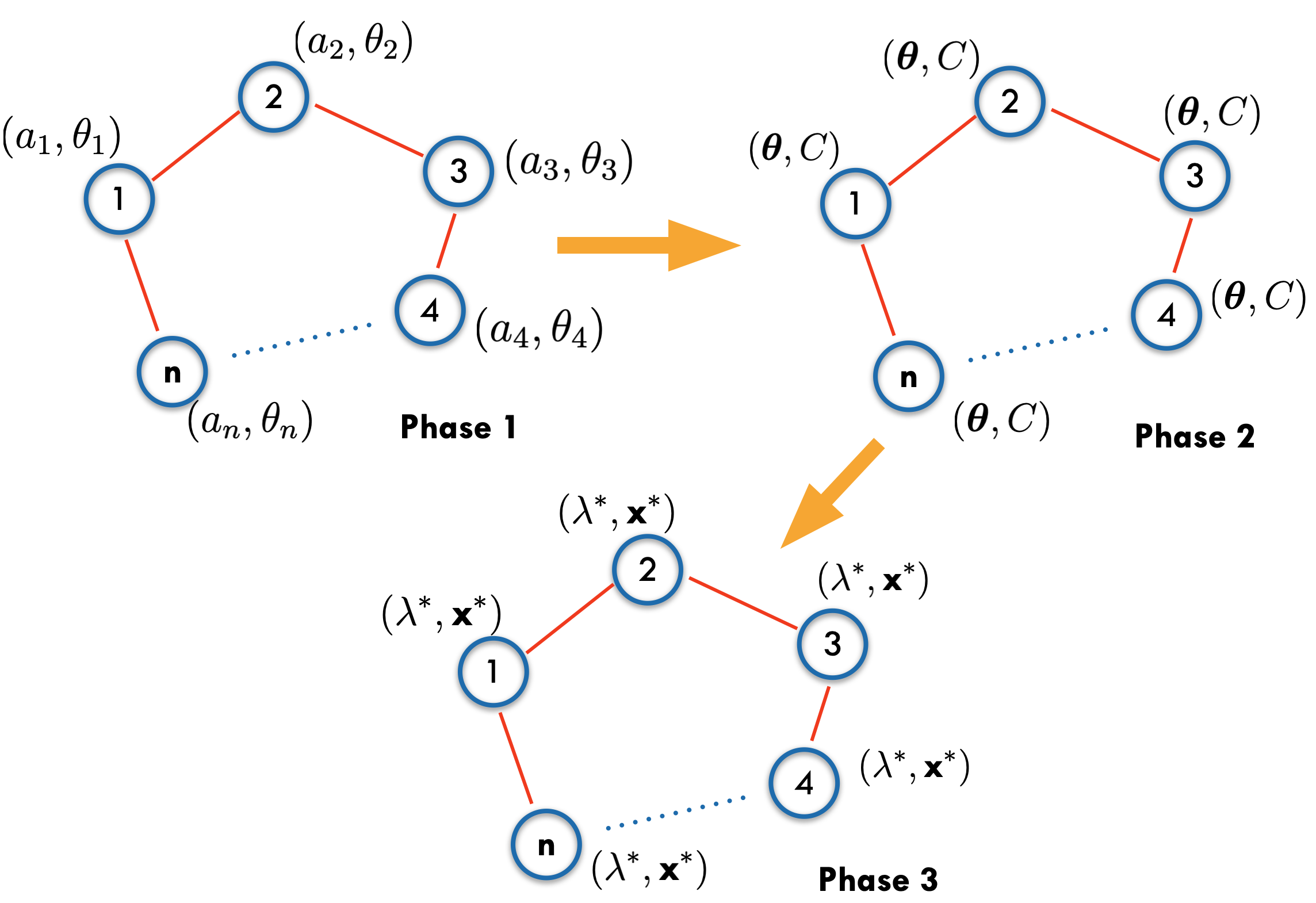}
  \end{minipage}
  \caption{Centralized (left) and distributed (right) implementations of competitive equilibriums for the proposed transactive energy systems. Each dark blue circle represents an agent in the system; the brown circle represents a central coordinator agent; each red line represent a communication link. The implementations take place in three sequential phases.  } \label{fig-implementation}
\end{figure*}

Equivalence for the MTES refers to a direct mapping of the competitive equilibrium pair ($\lambda^\ast$, $\mathbf{x}^\ast$), to the social welfare equilibrium $\mathbf{x}^\ast$.  Or otherwise, if $\mathbf{x}^\star$ is a social welfare  equilibrium for the  MTES, then there exists $\lambda^\ast \in \mathbb{R}$ such that ($\lambda^\ast$, $\mathbf{x}^\star$) is a competitive equilibrium. Equivalence for MTES-ST is similarly defined, providing a guarantee that market-based agent decisions coincide with utility-based planner decisions. More specifically, by solving the optimization problem for the social welfare equilibrium for either MTES or MTES-ST, the optimal dual variable is a Lagrangian multiplier associated with the supply-demand balance constraint, which is the optimal price for the competitive equilibrium   \cite{CDC}. 

Importantly, there is a critical difference between MTES and MTES-ST in terms of optimal pricing. For MTES, the prices under a competitive equilibrium can be either positive or negative. In contrast, for MTES-ST the price under a competitive equilibrium is always non-negative. Next, we explore connections between the two MTES and MTES-ST models. 

\begin{proposition}\label{prop-equivalence}
Suppose each $h(\cdot; \theta_i)$ is a concave function over the domain $\mathbb{R}^{\geq 0}$. Assume $\lambda_1^ \ast > 0$ and $\lambda_2^ \ast > 0$ correspond to optimal pricing signals for the competitive equilibrium for MTES and MTES-ST, respectively. Then there holds $\lambda_1^ \ast =\lambda_2^ \ast$,  and the agent decisions for MTES and MTES-ST are the same for the respective competitive equilibriums. 
\end{proposition}
\begin{proof}
    See Appendix \ref{Appendix_A1}.
\end{proof}

\subsection{Practical Implementations of Competitive Equilibriums}
In Fig.~\ref{fig-implementation} we illustrate two approaches to implementing MTES or MTES-ST, each approach including three sequential phases (i.e., P1, P2, P3). 

  \noindent \emph{Approach 1: Centralized computation via an aggregator}
\begin{itemize}
    \item [P1] The aggregator directly communicates with each agent $i \in \mathrm{V}$ to collect their individual load preferences $\theta_i$ and local power production $a_i$. 
    \item [P2] The aggregator computes  the overall network generation production $C$, then the social welfare equilibirum $\mathbf{x}^\ast$, and identifies the corresponding Lagrangian multiplier $\lambda^\ast$, which is the optimal price.
    \item [P3] The aggregator directly sends $(\lambda^\ast,\mathbf{x}^\ast)$ to each agent $i \in \mathrm{V}$.
    \end{itemize}
    
    \noindent \emph{Approach 2: Distributed computation without an aggregator} 
    \begin{itemize}
    \item[P1] Each agent $i\in \mathrm{V}$ selects their individual preference vector $\theta_i$, and communicates ($\theta_i, a_i$) to other agents. 
      \item [P2] Each agent $i\in \mathrm{V}$ runs, for example, a distributed processor agreement protocol \cite{1980} to receive vectors $\bm{\theta}=(\theta_1,\dots,\theta_n)^\top$ and $\mathbf{a}=(a_1,\dots,a_n)^\top$, and each agent computes the overall network generation production $C$.
    \item [P3] Each agent $i\in \mathrm{V}$ independently solves a social welfare equilibrium $\mathbf{x}^\ast$ and identifies the corresponding Lagrangian multiplier which is the optimal price $\lambda^\ast$. 
\end{itemize}
Each of the two implementations are underpinned by Proposition~1 that states the competitive equilibriums are equivalent to the social welfare equilibriums. That is, the resulting $\mathbf{x}^\ast$ is guaranteed to be a competitive equilibrium. Since the social welfare equilibriums depend only on the network generation production $C$, the processor agree protocol for the vector $\mathbf{a}$ can be replaced by a distributed average consensus algorithm \cite{Mesbahi2010}, where each agent obtains the average of all $a_i$, or equivalently, the network capacity $C$. 

\section{The Problem of Social Shaping}\label{rom3}
In this section, we motivate and define the problem of social shaping in a multi-agent transactive energy network. We focus on social shaping agent preferences to support optimization-based design of electricity prices. 
	
\subsection{Motivating Example}
Here, we provide a motivating example to highlight conditions where agent preferences significantly influence optimal pricing under competitive equilibriums.   
\vspace{2mm}

\noindent{\em Example 1}. Consider a MTES with four agents, where the electricity supply available is $\mathbf{a}=(a_{1},a_{2},a_{3},a_{4})^\top=(5,8,7,0)^\top.$ Each agent $i$ has a linear-quadratic utility function of the form $h(x_{i}; \theta_i)= -\frac{1}{2}b_{i}x_{i}^{2}+m_{i}b_{i}x_{i}$, where $(b_{1},b_{2},b_{3},b_{4})=(2,5,3,10)$ and $(m_{1},m_{2},m_{3},m_{4})=(6, 5, 6, 5)$. The social welfare equilibrium is computed by solving the optimization problem in \eqref{opt_social_LD_1}, which yields $\mathbf{x}^{\ast}=(5.12, 4.65, 5.41,4.82)^{\top}$. The Lagrangian multiplier corresponding to the supply-demand balance constraint $\sum_{i=1}^{n}x_{i} = \sum_{i=1}^{n}a_{i}$ is $\lambda^{\ast} = 1.765.$

Next, let $m_{4}$ take values in the interval $[5,30].$ We sample the interval uniformly with a step-size of $1$ to obtain $26$ different values for $m_{4}.$ For each $m_{4}$, the social welfare equilibrium and optimal price are computed. In Fig.~\ref{fig:motivating_example}, we present the optimal price $\lambda^{\ast}$ and the optimal resource allocation $x_{4}^{\ast}$ of agent $4$ as functions of $m_{4}$.  From Fig.~\ref{fig:motivating_example}, we observe that the optimal price $\lambda^{\ast}$ at $m_{4}=30$ exceeds the  optimal price at $m_{4}=5$ by a factor of $57$. Furthermore, we observe that the preferences of a single agent significantly alter the electricity price, and influence the distribution of resources allocations. \hfill$\square$

\begin{figure}
    \centering
    \includegraphics[width=0.35\textwidth]{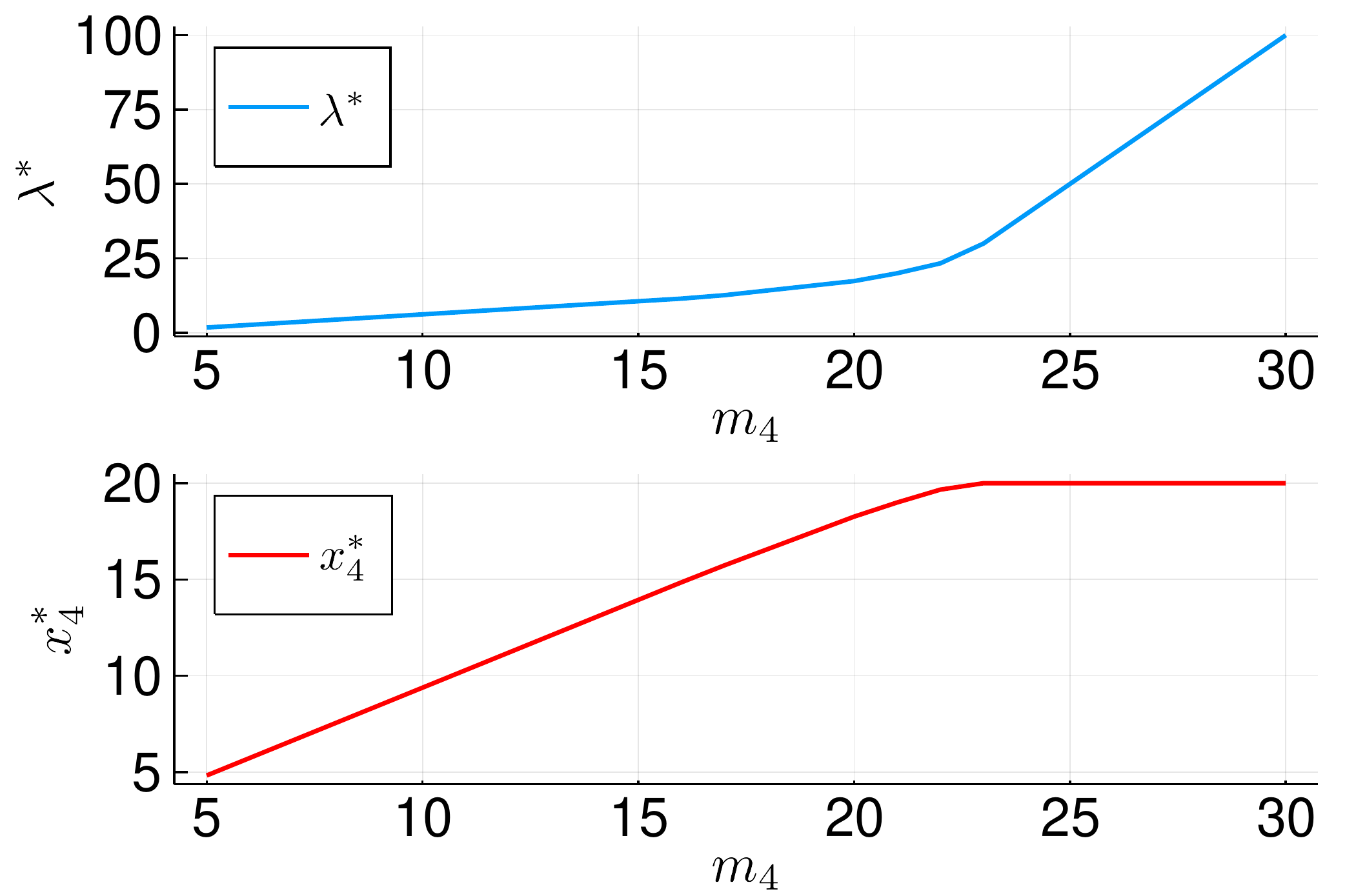}
    \caption{Top: optimal price $\lambda^{\ast}$ as functions of $m_{4}$. Bellow: agent~$4$ optimal resource allocation, $x_{4}^{\ast}$, as functions of $m_{4}$. }
    \label{fig:motivating_example}
\end{figure}

\subsection{Social Shaping Problem}
From Example 1, we observed that the optimal price for the MTES is in need of social shaping at a system level. That is, we observed that one agent was able to consume all available resources at a price that potentially prohibited others from accessing the energy. In what follows, we consider optimal pricing below an upper limit, where the upper limit is deemed socially acceptable by all agents. In this way, utility functions will be restricted within a prescribed range within our MTES framework. More importantly, socially admissible utility functions must correspond to socially acceptable optimal pricing. To this end, we define a social shaping problem for MTES and MTES-TS.

\begin{definition}(Social Shaping of Agent Preferences) Consider both MTES and MTES-ST. Let $\Theta$ be a space of parameters, where the preference $\theta_i \in \Theta$ is selected by agent $i$ in the context of their utility function $h(\cdot;\theta_i)$ for $i\in \mathrm{V}$. Let $\lambda^\dag >0$ denote an upper energy pricing limit, which represents the threshold of socially acceptable prices for agents $i\in \mathrm{V}$.  Find the set $\Theta$, such that agent preferences $(\theta_1,\dots,\theta_n) \in \Theta^n$ can be selected in a way where the optimal price corresponding to the competitive equilibriums satisfies $ \lambda^\ast \leq \lambda^\dag$. 
\end{definition}

\section{Conceptual Social Shaping Solvability} \label{rom4}

In this section, we study the solvability of the social shaping problem at conceptual levels.

\subsection{MTES}\label{conceptual_MTES}

Here, we introduce the social shaping problem for strictly concave and differentiable functions. Let each $h(\cdot;\theta_i)$ be continuously differentiable and strictly concave. It follows from Proposition~\ref{prop_equilibrium_equivalence} that the social welfare equilibrium and the competitive equilibrium exist and are equivalent. Consequently, we can consider either the social welfare problem or the competitive problem. 

Let $x_i^\ast$ be a solution to  (\ref{opt_LD_1}). Since $h(\cdot;\theta_i)$ is continuously differentiable and strictly concave, $h'(\cdot;\theta_i)$ is strictly monotone. We denote by $l(\cdot; \theta_i)$ the inverse function of $ h'(\cdot;\theta_i)$. Specifically, when $h(\cdot;\theta_i)$ is strictly concave, $ h'(\cdot;\theta_i)$ is strictly monotone, and as such the inverse must exist. As a result, we can derive 
	\begin{equation}\label{eq_x_MTES}
		x_i^\ast = \max \left\{l(\lambda^\ast; \theta_i), 0\right\}.
	\end{equation}
Next, substitute \eqref{eq_x_MTES} into the balancing equality in \eqref{load_demand_supply_constraints}, which yields
	\begin{equation}\label{eq_30}
		\sum_{i=1}^n \max \left\{l(\lambda^\ast; \theta_i), 0\right\} =C.
	\end{equation}
 Since $C>0$, there exists at least one agent with $x_i^\ast \neq 0$, corresponding to $x_i^\ast=l(\cdot; \theta_i)$. Also, as $h'(\cdot;\theta_i)$ is strictly monotone, its inverse $l(\cdot; \theta_i)$ is also strictly monotone. Consequently, the left-hand side of \eqref{eq_30} is the sum of at least one strictly monotone function, implying that the summation result is also a strictly monotone function whose inverse exists. Fixing $\bm{\theta}=(\theta_1,\dots,\theta_n)$, then there holds
\begin{equation}
    \lambda^\ast = \hat l (C; \bm{\theta}),
\end{equation}
where $\hat l (\cdot; \bm{\theta})$ is the inverse of the left-hand side of \eqref{eq_30}.

Next, we denote by $\chi_{\Theta}$ the \emph{maximal value of the set of optimal prices} that support all competitive equilibriums when agent preferences are drawn from $\bm{\theta}\in \Theta^n$. We define the maximal value of the set of optimal prices by
\begin{align}\label{chi_theta}
    \chi_{\Theta}:= \max_{\bm{\theta}\in \Theta^n} \hat l (C; \bm{\theta}).
\end{align}
From the definition of $\chi_{\Theta}$, the following result is immediately clear indicating the social shaping problem is conceptually captured by a set decision problem with respect to $\Theta$. 

\begin{theorem}\label{theorem6}
	Consider a MTES. Suppose each $h(\cdot;\theta_i)$ is strictly concave and differentiable over $\mathbb{R}^{\geq 0}$. Let $\lambda^\dag >0$ denote the threshold for socially acceptable energy pricing, such that $\lambda^\ast \leq \lambda^\dag$.  Then any set $\Theta$ satisfying $\chi_{\Theta}\leq \lambda^\dag$ solves the problem for social shaping of agent preferences. 
\end{theorem}

 \subsection{Homogenous MTES}

Consider the situation where agents maintain a common preference, i.e., $\theta_i=\hat{\theta}$. For example, each agent $i$ maintains a preference $\hat{\theta}$, which is an average of true preferences across all agents, that is $\hat{\theta}=\frac{1}{n}\sum_{i=1}^n {\theta}_i,\ i\in\mathrm{V}$, where the set $\Theta$ is convex. 

 \begin{theorem}
Consider a MTES with homogenous preferences, i.e., $\theta_i=\theta$ for all $i\in \mathrm{V}$.  Suppose each $h(\cdot;\theta_i)$ is concave and differentiable over $\mathbb{R}^{\geq 0}$. Let $\lambda^\dag >0$ denote the threshold for socially acceptable energy pricing, such that $\lambda^\ast \leq \lambda^\dag$. Then a solution for $\Theta$  for the social shaping problem is given by	\begin{equation}\label{eq_set_MTES_homogenous}
		\Theta = \big\{ \theta:  h' ({C}/{n}; \theta) \leq \lambda^\dag \big\}.
	\end{equation}
\end{theorem} 
\begin{proof}
See Appendix \ref{Appendix_A}.
\end{proof}

\subsection{MTES-ST}
 In our prior work \cite{CDC}, we showed that $\lambda^\ast \geq 0$ for MTES-ST. It follows that $\lambda^\ast = 0$ corresponds to socially resilient pricing, as it satisfies $\lambda^\ast \leq \lambda^\dag$. In what follows, we consider socially resilient pricing for MTES-ST for the case where $\lambda^\ast > 0$.
 
 According to Proposition~\ref{prop-equivalence}, when the price is positive, MTES and MTES-ST yield the same optimal pricing $\lambda^\ast >0$, and agent decisions. Consequently, we define $\chi_{\Theta}$ as \eqref{chi_theta}, and introduce the following theorem.
 

\begin{theorem} 
	Consider a MTES-ST. Suppose each $h(\cdot;\theta_i)$ is strictly concave and differentiable over $\mathbb{R}^{\geq 0}$. Let $\lambda^\dag >0$ denote the threshold for socially acceptable energy pricing, such that $\lambda^\ast \leq \lambda^\dag$. Then any set $\Theta$ satisfying $\chi_{\Theta}\leq \lambda^\dag$ solves the problem for social shaping of agent preferences. 
\end{theorem}

 \section{Analytical Social Shaping Solutions}\label{rom5}
 
In this section, we focus on two fundamental classes of utility functions, i.e., linear-quadratic functions and piece-wise linear functions, and we show the social shaping problems can be explicitly solved.

\subsection{Linear Quadratic Utility Functions}

We impose the following assumption.

\vspace{4mm}

\noindent{\em Assumption 1}. Let $\theta_i:=(b_i,m_i)\in \mathbb{R}^{> 0} \times \mathbb{R}^{> 0}$, where \begin{equation}\label{eq_utility_q}
	h(x_i; \theta_i)=-\frac{1}{2} b_i x_i^2 +m_i b_i x_i. 
	\end{equation} 
Fig. \ref{fig_utility_quadratic} illustrates example linear quadratic utility functions.
\begin{figure}[h]
  \centering
  \includegraphics[width=.75\linewidth]{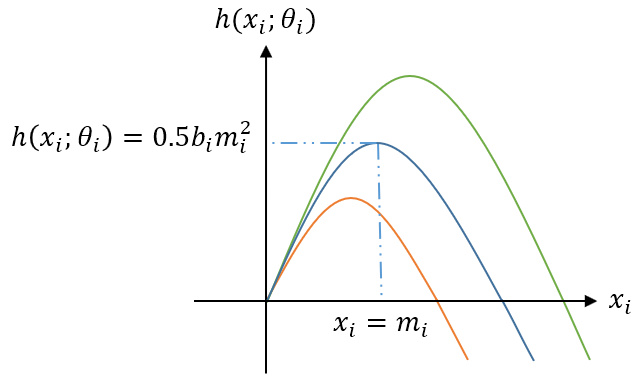}
  \caption{Example linear quadratic utility functions. }  \label{fig_utility_quadratic}
\end{figure}

\begin{theorem}\label{theorem1}
Consider a MTES with a quadratic utility function as defined in Assumption 1. Assume $\lambda^\dag$ is the upper limit of the socially acceptable price.  Suppose $(b_{\rm max}, m_{\rm max})\in \mathbb{R}^2_{> 0}$ is selected from the following set
	\begin{equation}\label{eq_set_1}
		\begin{gathered}
			\begin{aligned}
		\mathscr{S}_{\ast} =\left\{
		m_{\rm max} \leq \frac{C}{n}, \, b_{\rm max} \in \mathbb{R}^{>0} \right\} \mcup \\ \left\{ m_{\rm max} > \frac{C}{n}, \, b_{\rm max} \leq \frac{n \lambda^\dag}{nm_{\rm max}-C} \right\}.
		   \end{aligned}
		\end{gathered}
	\end{equation} 
Then  $\lambda^\ast$ is always socially acceptable since $\lambda^\ast \leq \lambda^\dag$, for all utility functions satisfying  $b_i \leq b_{\rm max}$ and $m_i \leq m_{\rm max}$.
\end{theorem}
\begin{proof}
See Appendix \ref{Appendix_B}.
\end{proof}

 
\noindent{\em Example 2}. Consider a MTES with linear-quadratic utility functions as defined in Assumption~1. 

Let the local resource $a_{i}$ of each agent $i$, be a random number sampled from a normal distribution $\mathcal{N}(\mu_{a},\,\sigma_{a}^{2})$ truncated to the interval $[0,10].$ The mean value of the local resources is $\mu_{a} = 5$, and the standard deviation is $\sigma_{a} = 1.25 $.  Overall network generation production, $C$, is the summation of all local resources, i.e., $C=\sum_{i=1}^{n}a_{i}.$ Recall, the personalized parameter $\theta_i=(b_{i}, m_{i})$ represents the preferences of agent $i$. Let preferences for agent $1$, $(b_{1}, m_{1})$, be represented by $m_{1}$, a value from a uniform distribution truncated to the interval $[\frac{C}{n},\frac{100C}{n}]$, and $b_{1} = \frac{n\lambda^{\dag}}{nm_{1}-C}$.  Given $(b_{1}, m_{1}),$ let the personalize parameters $\theta_i=(b_{i}, m_{i})$ for all other agents be sampled from two uniform distributions truncated to the interval $(0,b_{1}]$ and the interval $(0,m_{1}]$, respectively. In this manner, we have $(b_{\max}, m_{\max})=(b_{1}, m_{1})$, and in what follows we verify that the personalized parameter $\theta_i$ for each agent is within the range prescribed in \eqref{eq_set_1}.

Consider $n=10000$ agents, and let the socially acceptable threshold $\lambda^{\dag}$  take value from $\{20,22,24,26,28,30\}$. For each possible $\lambda^{\dag}$, we conduct $K=1000$ numerical experiments, for each of which we generate a parameter set $(\mathbf{b}^{(k)}, \mathbf{m}^{(k)})$, where $k=1,\dots,K$, by following the aforementioned parameter selection process. In each experiment $k$, we input parameters $(\mathbf{b}^{(k)}, \mathbf{m}^{(k)})$ to compute the optimal price, that is, the Lagrangian multiplier associated with the equity constraint $\sum_{i=1}^{n}x_{i} = \sum_{i=1}a_{i}$, as in \eqref{opt_social_LD_1}. In Fig.~\ref{quadratic_pricing}, we present optimal pricing for quadratic utility functions within the range \eqref{eq_set_1} under different upper limits $\lambda^{\dag} = 20,22,24,26,28,30,$ each corresponding to the respective $1000$ numerical experiments. 


Next, let $\lambda^{\dag}=20$ and let $n \in \{100,1000,10000,100000\}$. For each possible value of $n$, we conduct $K=1000$ experiments with the corresponding personalized parameter sets, $\theta_i$, obtained by applying the aforementioned process. In Fig.~\ref{quadratic_scale}, we present optimal pricing for quadratic utility functions within the range \eqref{eq_set_1}, where $1000$ experiments are conducted for each system scale $n = 100,1000,10000,100000$.

\begin{figure}[ht]
    \centering
    \includegraphics[width=0.35\textwidth]{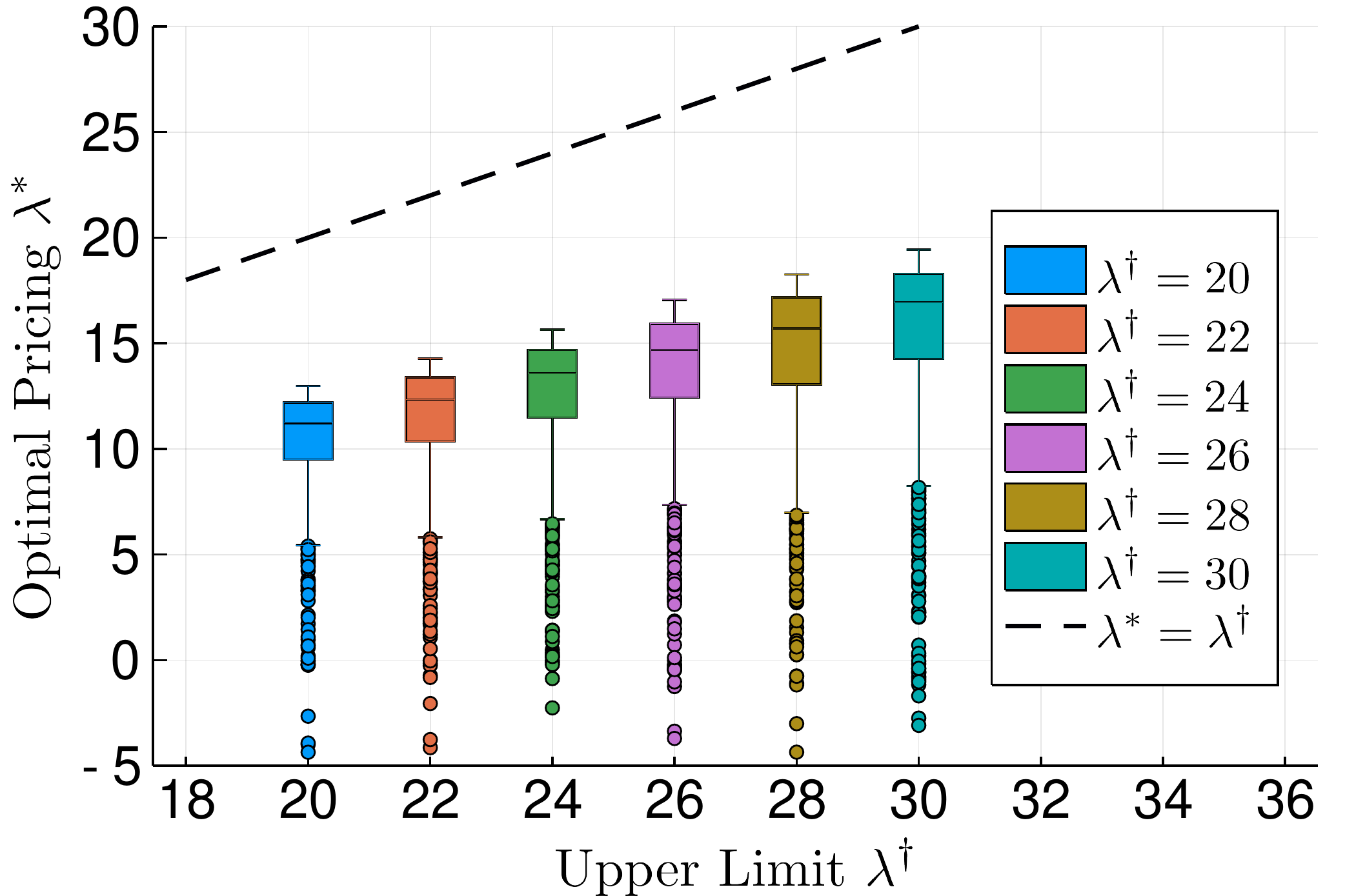}
    \caption{Optimal pricing for quadratic utility functions within the range \eqref{eq_set_1}, with upper electricity pricing limits of $\lambda^{\dag} = 20,22,24,26,28,30.$ }
    \label{quadratic_pricing}
\end{figure}

\begin{figure}[ht]
    \centering
    \includegraphics[width=0.35\textwidth]{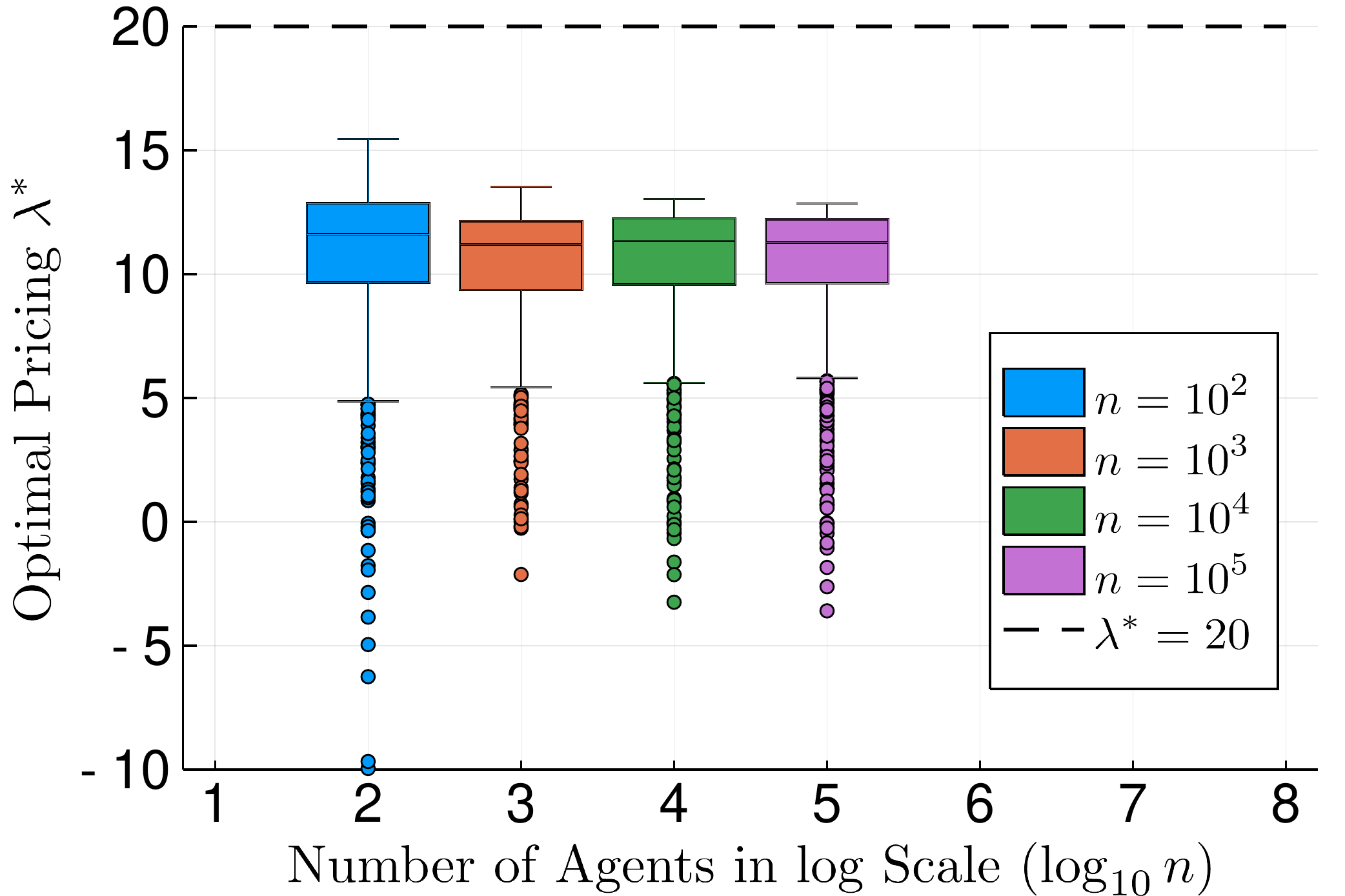}
    \caption{Optimal pricing for quadratic utility functions within the range \eqref{eq_set_1}, considering multi-agent systems of scale $n = 100,1000,10000,100000$.}
    \label{quadratic_scale}
\end{figure}

In Fig.~\ref{quadratic_pricing} and Fig.~\ref{quadratic_scale}, the middle line, and the lower and upper boundaries of each box (interquartile range or IQR) correspond to the median, and 25/75 percentile of the $1000$ optimal prices, respectively. The lower and upper whiskers extend maximally $1.5 \times$ of IQR from 25 percentile downwards and 75 percentile upwards, respectively. The points that are located outside the whiskers are considered data outliers. 

In Fig~\ref{quadratic_pricing}, we observe that the optimal prices in the numerical experiments are below the corresponding upper limit, i.e., $\lambda^{\ast} = \lambda^{\dag}$, indicating the optimal pricing being socially acceptable. In Fig.~\ref{quadratic_scale}, we observe that the optimal prices in the numerical experiments considering various multi-agent system scales, are lower than the upper limit $\lambda^{\dag} = 20$. Our observations in both Fig.~\ref{quadratic_pricing} and Fig.~\ref{quadratic_scale} correspond to, and validate Theorem~\ref{theorem1}.  \hfill$\square$

\subsection{Piece-wise Linear Utility Functions}

We next consider  the following assumption. 
\vspace{2mm}

\noindent{\em Assumption 2}. Let $\theta_i:=(\beta_i, \phi_i)\in \mathbb{R}^{> 0} \times \mathbb{R}^{> 0}$ with	
	\begin{equation}\label{eq_utility_linear}
h(x_i; \theta_i):   = \min \{\beta_i x_i, \phi_i \beta_i\}.\end{equation}
Fig. \ref{fig_utility_linear} illustrates example piece-wise linear utility functions.
	 
 \begin{figure}[h]
  \centering
  \includegraphics[width=.7\linewidth]{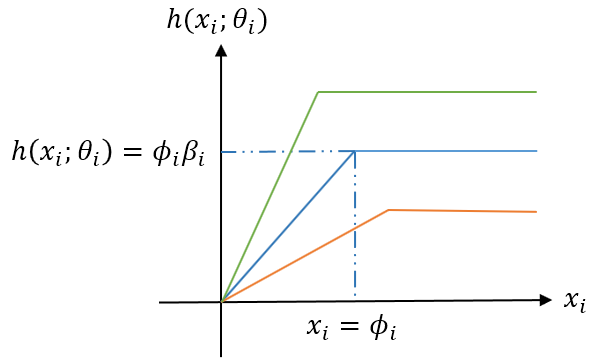}
  \caption{The class of piece-wise linear functions. }  \label{fig_utility_linear}
\end{figure}

\begin{theorem}\label{theorem4}
	Consider a MTES with piece-wise linear utility functions as defined in Assumption 2.  Suppose $(\beta_{\rm max}, \phi_{\rm max}) $ is selected from the following set
	\begin{equation}\label{eq_set_3}
		\begin{gathered}
			\begin{aligned}
				\mathscr{S}_{\ast} =\left\{
				\phi_{\rm max} < \frac{C}{n}, \, \beta_{\rm max} \in \mathbb{R}^{>0} \right\} \mcup \\ \left\{ \phi_{\rm max} \geq \frac{C}{n}, \, \beta_{\rm max} \leq \lambda^\dag \right\}.
			\end{aligned}
		\end{gathered}
	\end{equation} 
Then the resulting  $\lambda^\ast$ is always socially acceptable since $\lambda^\ast \leq \lambda^\dag$, for all utility functions satisfying  $\beta_i \leq \beta_{\rm max}$ and $\phi_i \leq \phi_{\rm max}$.
\end{theorem}
\begin{proof}
See Appendix \ref{Appendix_C}.
\end{proof}


\noindent{\em Example 3}.  Consider a MTES where each agent $i$ is associated with a piece-wise linear utility function as defined in Assumption 2. 

The local resources ${a_i}$ of each agent is obtained following the same process described in Example 2.  The network generation production, $C$, is then, $C=\sum_{i=1}^{n}a_{i}.$ Let preferences for agent $1$, $(\beta_{1}, \phi_{1})$, be represented by $\beta_{1}=\lambda^{\dag}$ and let $\phi_{1}$ be a random number sampled from a normal distribution truncated to the interval $[\frac{C}{n}, \frac{10C}{n}]$. Let the personalize parameter pairs for the remaining agents be sampled from two uniform distributions, truncated to the interval $(0,\beta_{1}]$ and the interval $(0,\phi_{1}]$, respectively. In this way, we seek to validate the design of $(\beta_{\max}, \phi_{\max})=(\beta_{1}, \phi_{1})$ in the context of the range in \eqref{eq_set_3}.
		
Let $n=10000$ denote the number of agents, and let the upper limit $\lambda^{\dag}$ take values from $ \{20,22,24,26,28,30\}.$ For each upper limit $\lambda^{\dag}$, we carry out  $K=1000$ experiments. In each experiment $k =1,\dots,K$, a different parameter set of $(\bm{\beta}^{(k)}, \bm{\phi}^{(k)})$ is obtained upon the aforementioned parameter generation process. For each $(\bm{\beta}^{(k)}, \bm{\phi}^{(k)})$, we solve the social welfare optimization problem ~\eqref{opt_social_LD_1}, and the optimal dual variable corresponding to the equity constraint $\sum_{i=1}^{n}x_{i}=\sum_{i=1}^{n}a_{i}$ is obtained as $\lambda^{\ast}$. In Fig.~\ref{linear_piecewise_pricing}, we present the optimal pricing for piece-wise linear utility functions within the range \eqref{eq_set_3}, considering $1000$ numerical experiments for each upper limit $\lambda^{\ast} = 20,22,24,26,28,30$.  From Fig.~\ref{linear_piecewise_pricing}, we observe that the dotted line $\lambda^{\ast} = \lambda^{\dag}$, sits above all optimal prices identified from the numerical experiments, corresponding to the optimal price being socially acceptable. 

\begin{figure}[ht]
    \centering
    \includegraphics[width=0.35\textwidth]{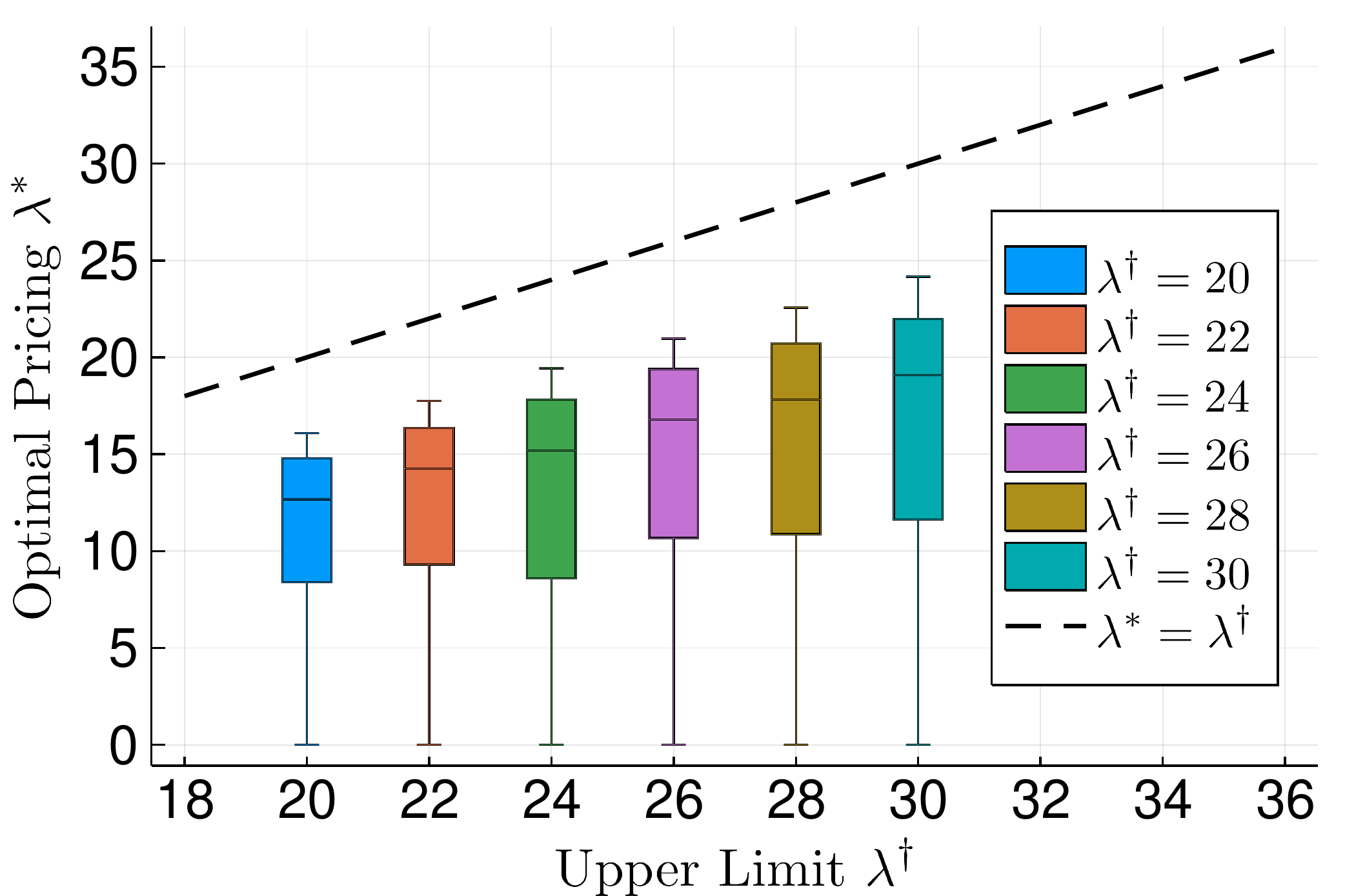}
    \caption{Optimal pricing for piece-wise linear utility functions within the range \eqref{eq_set_3}, considering upper limits $\lambda^{\ast} = 20,22,24,26,28,30.$} 
    \label{linear_piecewise_pricing}
\end{figure}

\begin{figure}[ht]
    \centering
    \includegraphics[width=0.35\textwidth]{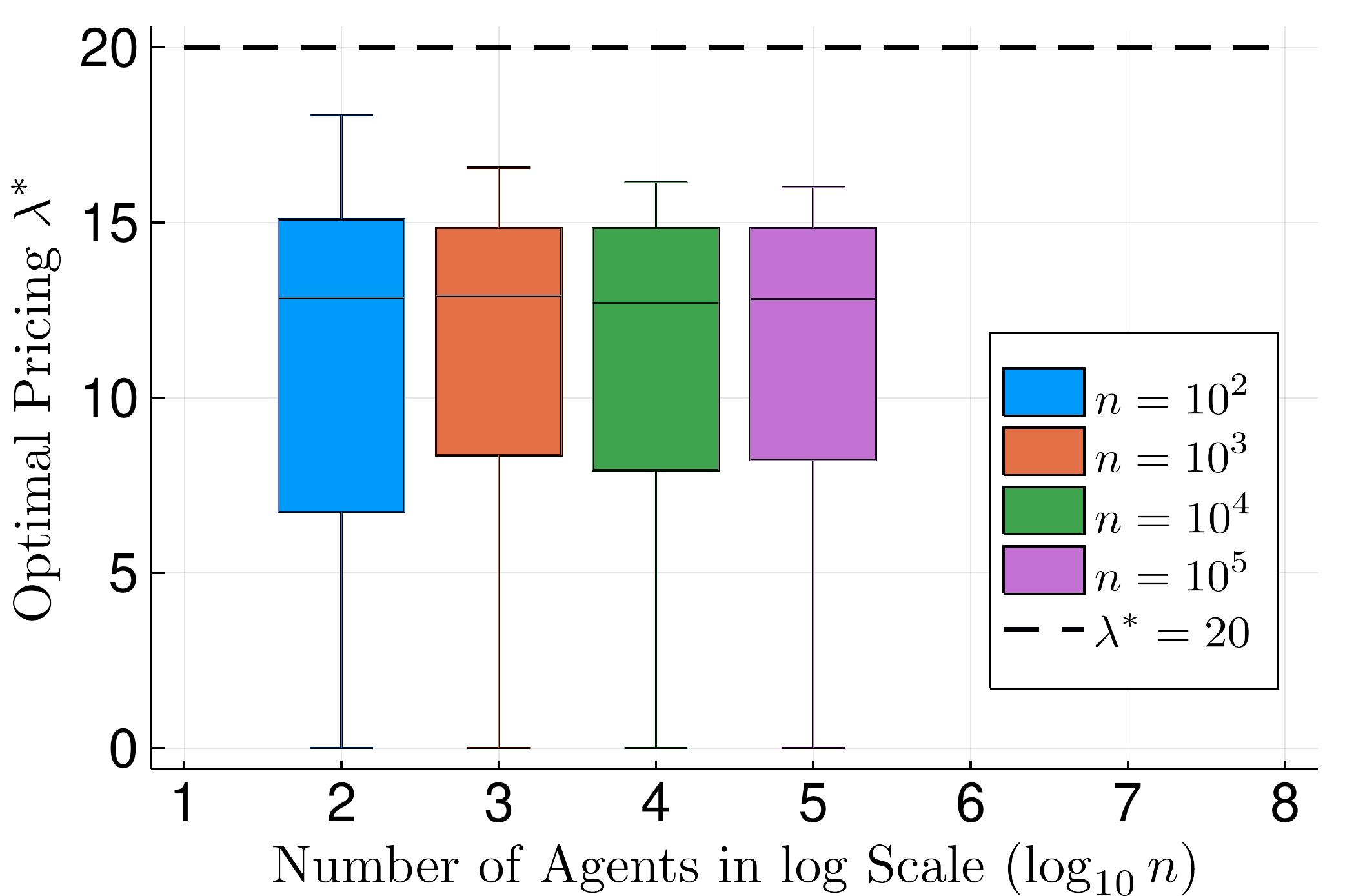}
    \caption{Optimal pricing for piece-wise linear functions within the range \eqref{eq_set_3}, for multi-agent systems of scale $n=100, 1000, 10000, 100000$.}
    \label{linear_piecewise_scale}
\end{figure}

Next, let $\lambda^{\dag}=20$ and let $n \in \{100,1000,10000,100000\}$. For each possible value of $n$, we conduct $K=1000$ experiments with the corresponding personalized parameter sets, $\theta_i$, obtained by applying the aforementioned process. We then solve the optimization problem~\eqref{opt_social_LD_1} and obtain the optimal price $\lambda^{\ast}$. In Fig.~\ref{linear_piecewise_scale}, we present the optimal prices for piece-wise linear functions within the range \eqref{eq_set_3}, considering the respective $1000$ experiments, for each system scale $n=100, 1000, 10000, 100000$. In Fig.~\ref{linear_piecewise_scale}, we observe that all optimal prices in the numerical experiments, considering various system scales, are also below the upper limit $\lambda^{\dag}=20$ required for social acceptance by the agents. From Fig.~\ref{linear_piecewise_pricing} and Fig.~\ref{linear_piecewise_scale}, we observe that piece-wise linear utility functions within the range \eqref{eq_set_3}, correspond to, and validate Theorem~\ref{theorem4}. \hfill$\square$

\subsection{Further Discussions: MTES-ST}
It is worth emphasizing the results in  Theorem \ref{theorem1} and Theorem \ref{theorem4} for MTES continue to be valid for MTES-ST. That is, the same set of preference parameters will enable socially acceptable optimal pricing under competitive equilibriums for MTES-ST. Specifically, from Proposition \ref{prop-equivalence},  when the price is {positive},  MTES and MTES-ST are exactly the same in terms of pricing and agent decisions. This equivalence is related to the second part of the sets proposed in \eqref{eq_set_1} and \eqref{eq_set_3}, i.e, $m_{\rm max} > \frac{C}{n}$ and $\phi_{\max} \geq \frac{C}{n}$, respectively. 

{Regarding the first part of the proposed sets, i.e, $m_{\rm max} \leq \frac{C}{n}$ or $\phi_{\max} < \frac{C}{n}$,   Lemmas \ref{lemma2} and \ref{lemma9}  in Appendix \ref{Appendix_D} show that under these conditions the optimal price for MTES-ST is equal to zero, and thus socially acceptable.} 

\section{Conclusions}\label{rom6}
This paper has considered the problem of shaping agent utility functions in a transactive energy system. In the overall system design, energy supply and demand is balanced while incorporating individual utility preferences for agents, in a way that accommodates socially acceptable pricing. In defining the social shaping problem, we identified a set of preference parameters to provide guarantees that the optimal energy price is below a prescribed threshold. Our established conceptual results indicated that the social shaping problem is characterized by a set decision problem. We also presented two analytical solutions in tight ranges for coefficients of  linear-quadratic utilities and piece-wise linear utilities, to further demonstrate the application of the social shaping problem to socially acceptable electricity pricing. Future work might include extensions of the social shaping ideas to transactive energy systems operating dynamically, or systems that operate under game-theoretic frameworks. 

\section*{Acknowledgements}
This work was supported by the Australian Research Council under Grants DP180101805 and DP190103615.


\section{Appendix}

\subsection{Proof of Proposition 2}\label{Appendix_A1}
    It is assumed $h(\cdot;\theta_i)$ is concave. Therefore, Proposition~\ref{prop_equilibrium_equivalence} states that the social welfare equilibrium and the competitive equilibrium coincide. Consequently, either the social welfare problem or the competitive problem can be solved.
	Here, we consider the competitive optimization problem of MTES-ST in \eqref{opt_LTD_1}.
	The inequality constraint $x_i + e_i \leq a_i$ can be written as
	\begin{equation}\label{eq15}
		x_i + e_i + s_i = a_i,
	\end{equation}
	where $s_i \geq 0$ is the slack variable. Additionally, substituting $e_i = a_i - x_i - s_i$ into \eqref{opt_LTD_1} yields an equivalent form for the optimization problem as 
	\begin{equation}\label{opt_LTD_4}
		\begin{aligned}
			\max_{{x}_i, s_i} \quad &  h(x_i; \theta_i)+\lambda^\ast_2 (a_i - x_i -s_i) \\
			\quad & x_i, s_i\in \mathbb{R}^{\geq0}.
		\end{aligned}
	\end{equation}
	Let $\lambda^\ast_2 >0$. Then, the objective function in \eqref{opt_LTD_4} is strictly decreasing with respect to $s_i$. Consequently, in order for \eqref{opt_LTD_4} to be maximized, $s_i$ must be minimized, i.e., $s_i^\ast =0$. This implies that the inequality constraint $x_i + e_i \leq a_i$ is active and  
	\begin{equation}\label{eq21}
		x_i^\ast + e_i^\ast = a_i.
	\end{equation}
	Taking the summation in \eqref{eq21} implies
	\begin{equation}\label{eq22}
		\sum_{i=1}^{n}x_i^\ast + \sum_{i=1}^{n}e_i^\ast = C,
	\end{equation}
	where $C= \sum_{i=1}^{n}a_i$. Then, substituting the balancing equality $\sum_{i=1}^{n} e_i^\ast =0$ in \eqref{trading_demand_supply_constraints} into \eqref{eq22} yields
	\begin{equation}\label{eq23}
		\sum_{i=1}^n x_i^\ast =C.
	\end{equation} 
	Furthermore, substituting $s_i=0$ into \eqref{opt_LTD_4} yields
	\begin{equation}\label{opt_LTD_5}
		\begin{aligned}
			\max_{{x}_i} \quad &  h(x_i; \theta_i)+\lambda^\ast_2 (a_i - x_i) \\
			\quad & x_i\in \mathbb{R}^{\geq0}.
		\end{aligned}
	\end{equation}
	Comparing \eqref{opt_LTD_5} and \eqref{eq23} with \eqref{opt_LD_1} and \eqref{load_demand_supply_constraints} implies that the problem of MTES-ST is equivalent to the problem of MTES. Therefore, the agent decisions for MTES and MTES-ST are the same for the respective competitive equilibriums.  A similar analysis can be done to show $\lambda_{1}^\ast = \lambda_{2}^\ast$.

\subsection{Proof of Theorem 2}\label{Appendix_A}
Since $h(\cdot;\theta_i)$ is concave, Proposition~\ref{prop_equilibrium_equivalence} holds. Therefore, either the social welfare problem or the competitive problem can be examined. Here, we consider the competitive optimization problem of MTES in \eqref{opt_LD_1}. The optimal solution $x_i^\ast$ solves the following optimization problem
\begin{equation}\label{opt_MTES_homogenous}
	\begin{aligned}
		\max_{{x}_i} \quad &  h(x_i; \theta) +  \lambda^\ast(a_i - x_i) \\
		{\rm s.t.} \quad & x_i\in \mathbb{R}^{\geq0}.
	\end{aligned}
\end{equation}
Let $\hat x_i$ be the value $x_i$ which maximizes the objective function in the absence of any constraints. Then $\hat x_i$ is obtained when the derivative of the objective function equals zero, i.e, 
\begin{equation}\label{eq_gradient_MTES_homogenous}
	 h' (\hat x_i; \theta) = \lambda^\ast.
\end{equation}
This implies that all agents have the same $\hat x_i = \hat x$. Considering the inequality constraint $x_i\in \mathbb{R}^{\geq0}$ in \eqref{opt_MTES_homogenous}, if $\hat x \leq 0$ then $x_i^\ast =0$ for all agents, which contradicts the balancing equality $\sum_{i=1}^{n} x_i^\ast = C$ in \eqref{load_demand_supply_constraints} (note that $C>0$). Consequently, it follows  $\hat x >0$ and $x_i^\ast = \hat x$, which is positive and satisfies the inequality constraint. Additionally, from the balancing equality $\sum_{i=1}^{n} x_i^\ast = C$, we yield $x_i^\ast = \hat x = \frac{C}{n}$. Then, substituting $\hat x = \frac{C}{n}$ into \eqref{eq_gradient_MTES_homogenous} obtains
\begin{equation}\label{eq33}
	 h' (C/n; \theta) = \lambda^\ast.
\end{equation}
If $\theta$ is selected in such a way that $ h' (C/n; \theta) \leq \lambda^\dag$, then we yield $\lambda^\ast \leq \lambda^\dag$. This leads to the set $\Theta$ in \eqref{eq_set_MTES_homogenous}.

\subsection{Proof of Theorem~4}\label{Appendix_B}

\subsubsection{Preliminary Lemmas}
We first introduce some preliminary lemmas which are essential for the proof Theorem \ref{theorem1}.
\begin{lemma}\label{lemma4}
Consider the MTES with the quadratic utility function defined in Assumption 1. The optimal load allocation $x_i^\ast$, which is an optimal solution of the optimization problem \eqref{opt_LD_1}, is such that 
	\begin{equation}\label{eq_x}
		x_i^\ast = \max \left\{ m_i - \frac{\lambda^\ast}{b_i}, 0\right\}.
	\end{equation}
\end{lemma}
\begin{proof}
	Rearranging the optimization problem \eqref{opt_LD_1}, $x_i^\ast$ is the solution to the following maximization problem:
	\begin{equation}\label{opt_LD_2}
		\begin{aligned}
			\max_{{x}_i} \quad &  -\frac{1}{2} b_i x_i^2 + (m_i b_i - \lambda^\ast) x_i + \lambda^\ast a_i \\
			{\rm s.t.} \quad & x_i\in \mathbb{R}^{\geq0}.
		\end{aligned}
	\end{equation}
	Let $\hat x_i$ be the value $x_i$ which maximizes the objective function in the absence of any constraints. Then $\hat x_i$ is obtained when the derivative of the objective function equals zero. That is,
	\begin{equation}
		-b_i \hat x_i + (m_ib_i - \lambda^\ast) = 0,
	\end{equation}
	which implies $\hat x_i = m_i - \frac{\lambda^\ast}{b_i}$. Considering the inequality constraint $x_i \geq 0$ in the maximization problem \eqref{opt_LD_2}, when $\lambda^\ast \leq m_i b_i$, the solution  is achieved at $x_i^\ast = \hat x_i = m_i - \frac{\lambda^\ast}{b_i}$ which is non-negative and satisfies the inequality constraint. Conversely, when $\lambda^\ast > m_i b_i$, then $\hat x_i$ is negative which does not satisfy the inequality constraint, so $x_i^\ast \neq \hat x_i$. In this case, the objective function is strictly decreasing with respect to $x_i$. Consequently, in order for the objective function to be maximized, $x_i$ must be minimized, i.e., $x_i^\ast = 0$. Therefore, when $\lambda^\ast > m_i b_i$ then $x_i^\ast =0$, otherwise, $x_i^\ast = m_i - \frac{\lambda^\ast}{b_i}$.
\end{proof}

\begin{lemma}\label{lemma1}
	Consider the MTES with the quadratic utility function defined in Assumption 1. If $\sum_{i=1}^{n}m_i \leq C$,  then $\lambda^\ast \leq 0$. Conversely, if $\sum_{i=1}^{n}m_i > C$ then $\lambda^\ast >0$.
\end{lemma}
\begin{proof}
	(i) Consider the case $\sum_{i=1}^{n}m_i \leq C$. By contradiction, suppose $\lambda^\ast >0$. From equation \eqref{eq_x} we yield $x_i^\ast < m_i$, and therefore, $\sum_{i=1}^{n}x_i^\ast < \sum_{i=1}^{n}m_i$. Since $\sum_{i=1}^{n}m_i \leq C$, we obtain 
	\begin{equation}
		\sum_{i=1}^{n} x_i^\ast < C,
	\end{equation}
	which contradicts the balancing equality $\sum_{i=1}^n x_i^\ast = C$ in \eqref{load_demand_supply_constraints}. Therefore, it follows that $\lambda^\ast \leq 0$.
	
	(ii) Consider the case $\sum_{i=1}^{n}m_i > C$. By contradiction, suppose $\lambda^\ast \leq 0$. From equation \eqref{eq_x} we obtain $x_i^\ast \geq m_i$, and therefore, $\sum_{i=1}^{n}x_i^\ast \geq \sum_{i=1}^{n}m_i$. Since $\sum_{i=1}^{n}m_i > C$, we yield
	\begin{equation}
		\sum_{i=1}^{n} x_i^\ast > C,
	\end{equation}
	which contradicts the balancing equality $\sum_{i=1}^n x_i^\ast = C$ in \eqref{load_demand_supply_constraints}. Therefore, it follows that $\lambda^\ast > 0$.
\end{proof}

Consider two vectors $k=(k_1, ..., k_n)$ and 
$k'=(k'_1, ..., k'_n)$, and let $k \preceq k'$ denote $k_i \leq k_i'$ for all $i \in V$. Let $m=(m_1, ..., m_n)$ and $b=(b_1, ..., b_n)$. Suppose $\lambda^\ast$ is the optimal price associated with the pair of vectors $(m, b)$, and let $\lambda^{\ast'}$ be the optimal price associated with $(m', b')$.

\begin{lemma}\label{lemma3}
	Consider the MTES with the quadratic utility function defined in Assumption 1. If $\lambda^\ast > 0$, then $m \preceq m'$ and $b \preceq b'$ yield $\lambda^\ast \leq \lambda^{\ast'}$.
\end{lemma}
\begin{proof}
	Suppose $\lambda^\ast > 0$. Substituting  \eqref{eq_x} into the balancing equality $\sum_{i=1}^n x_i^\ast = C$ in \eqref{load_demand_supply_constraints} yields
	\begin{equation}\label{eq2}
		\sum_{i=1}^n \max \left\{ m_i - \frac{ \lambda^\ast}{b_i}, 0\right\} =C.
	\end{equation} 
	As $m_i$ and $b_i$ increase, $\lambda^\ast$ must also increase so as to compensate for the change --- ensuring the balancing equality \eqref{eq2} holds. Otherwise, the left-hand side of equality \eqref{eq2}  would increase, while the right-hand side remains constant, and so the equality would not hold.
\end{proof} 
 
\subsubsection{Proof of the Theorem}

Now, we present the proof of Theorem \ref{theorem1}.
	The quadratic utility function in Assumption 1 is concave, so Proposition~\ref{prop_equilibrium_equivalence} holds. We investigate two cases. 
			
	Case (i) $m_{\rm max} \leq \frac{C}{n}$. In this case, $m_i \leq m_{\rm max}$ implies $m_{i} \leq \frac{C}{n}$ for $i \in V$,  and therefore, $\sum_{i=1}^{n}m_i \leq C$. Consequently, Lemma \ref{lemma1} implies $\lambda^\ast \leq 0$. Since $\lambda^\dag >0$, one obtains $\lambda^\ast < \lambda^\dag$. Therefore, $\lambda^\ast$ is socially resilient. 
		
	Case (ii) $m_{\rm max} > \frac{C}{n}$ and $b_{\rm max} \leq \frac{n \lambda^\dag}{nm_{\rm max}-C}$. If $\lambda^\ast \leq 0$, then it is socially resilient. Conversely,
	if $\lambda^\ast >0$, Lemma \ref{lemma3} yields $\lambda^\ast$ is monotonically increasing with respect to $m_i$ and $b_i$, so the highest possible price $\lambda^\ast_{\rm max}$ is achieved when $m_i=m_{\rm max}$ and $b_i = b_{\rm max}$ for all agents $i \in V$. Consequently, when all agents select $m_i= m_{\rm max}$ and $b_i =  b_{\rm max}$, the balancing equality \eqref{eq2} results in
	\begin{equation}
		n \left( m_{\rm max} - \frac{\lambda^\ast_{\rm max}}{ b_{\rm max}}\right) = C,
	\end{equation}
	and therefore,
	\begin{equation}\label{eq3}
		\lambda^\ast_{\rm max} = b_{\rm max} \left( \frac{nm_{\rm max} -C}{n}\right).
	\end{equation}
	From equation \eqref{eq3}, along with the assumption $ b_{\rm max} \leq \frac{n \lambda^\dag}{nm_{\rm max}-C}$ in \eqref{eq_set_1}, yields $\lambda^\ast_{\rm max} \leq \lambda^\dag$. Since $\lambda^\ast \leq \lambda^\ast_{\rm max}$, one obtains $\lambda^\ast \leq \lambda^\dag$. 
		
	Considering  (i) and  (ii), it follows that as long as  $(b_{\rm max}, m_{\rm max})$ is constrained in the set $\mathscr{S}_{\ast}$ in \eqref{eq_set_1}, $\lambda^\ast$ will be socially resilient.

\subsection{Proof of Theorem 5}\label{Appendix_C}

\subsubsection{Preliminary Lemmas}
First, we provide some lemmas which are necessary for the proof of Theorem \ref{theorem4}.
\begin{lemma}\label{lemma6}
	Consider the MTES. If Assumption 2 holds, then $\lambda^\ast \geq 0$.
\end{lemma}

\begin{proof}
Note that $h (x_i; \theta_i)$ is a non-decreasing concave function. Therefore, according to Proposition~1 in \cite{CDC}, we yield $\lambda^\ast \geq 0$.	
\end{proof}

\begin{lemma}\label{lemma10}
Consider the MTES with the piece-wise linear utility function defined in Assumption 2. The optimal load allocation $x_i^\ast$, which is an optimal solution of the optimization problem \eqref{opt_LD_1}, satisfies 
	\begin{equation}\label{eq_x_linear}
		x_i^* = \left\{ {\begin{array}{ll}
				\left\{ {x: x \geq {\phi_i}} \right\} & \mbox{if $\lambda^\ast =0$,}\\
				{\phi_i} & \mbox{if $0<\lambda^\ast < \beta_i$,}\\
				\left\{ {x:0 \le x \le {\phi_i}} \right\} & \mbox{if $\lambda^\ast = \beta_i$,}\\
				0 & \mbox{if $\lambda^\ast > \beta_i$.}
		\end{array}} \right. 
	\end{equation}
\end{lemma}

\begin{proof}
	We investigate four cases.
	
	Case (i) $\lambda^\ast = 0$. In this case, the objective function in \eqref{opt_LD_1} equals the utility function $h(x_i; \theta_i) = \min \{\beta_i x_i, \phi_i \beta_i\}$, shown in Fig. \ref{fig_utility_linear}. As can be seen, the objective function is strictly increasing in the interval $x_i \in [0, \phi_i]$, while constant in the interval $x_i \in [\phi_i, \infty)$. Therefore, the optimal solution is achieved at $x_i^\ast \geq \phi_i$.
	
	Case (ii) $0<\lambda^\ast < \beta_i$. In this case, the objective function is strictly increasing in the interval $x_i \in [0, \phi_i]$, while strictly decreasing in the interval $x_i \in [\phi_i, \infty)$. Therefore, the optimal solution is achieved at $x_i^\ast=\phi_i$.
	
	Case (iii) $\lambda^\ast = \beta_i$. In this case, the objective function is constant in the interval $x_i \in [0, \phi_i]$, while strictly decreasing in the interval $x_i \in [\phi_i, \infty)$. Therefore, the optimal solution is achieved at $0 \leq x_i^\ast \leq \phi_i$.
	
	Case (iv) $\lambda^\ast > \beta_i$. In this case, the objective function is strictly decreasing in the whole interval $x_i \in [0, \infty)$. Therefore, the optimal solution is achieved at $x_i^\ast =0$.
\end{proof}

\begin{lemma}\label{lemma8}
	Consider the MTES with the piece-wise linear utility function defined in Assumption 2. If $\sum_{i=1}^{n}\phi_i < C$, then $\lambda^\ast = 0$. Conversely, if $\sum_{i=1}^{n}\phi_i > C$ then $\lambda^\ast > 0$.
\end{lemma}
\begin{proof}
	(i) Consider $\sum_{i=1}^{n}\phi_i < C$. By contradiction, suppose $\lambda^\ast>0$. According to \eqref{eq_x_linear}, $\lambda^\ast >0$ yields $x_i^\ast \leq \phi_i$ for $i \in V$, and therefore, $\sum_{i=1}^{n}x_i^\ast \leq \sum_{i=1}^{n}\phi_i$. Since $\sum_{i=1}^{n}\phi_i < C$, we yield
	\begin{equation}
		\sum_{i=1}^{n} x_i^\ast < C,
	\end{equation}
	which contradicts the balancing equality $\sum_{i=1}^n x_i^\ast = C$ in \eqref{load_demand_supply_constraints}. Therefore, it follows that $\lambda^\ast =0$. 
	
	(ii) Consider $\sum_{i=1}^{n}\phi_i > C$. By contradiction, suppose $\lambda^\ast=0$. According to \eqref{eq_x_linear}, $\lambda^\ast =0$ yields $x_i^\ast \geq \phi_i$ for $i \in V$, and therefore, $\sum_{i=1}^{n}x_i^\ast \geq \sum_{i=1}^{n}\phi_i$. Since $\sum_{i=1}^{n}\phi_i > C$, we yield
		\begin{equation}
		\sum_{i=1}^{n} x_i^\ast > C,
	\end{equation}
	which contradicts the balancing equality $\sum_{i=1}^n x_i^\ast = C$ in \eqref{load_demand_supply_constraints}. Therefore, it follows that $\lambda^\ast >0$. 
\end{proof}

\subsubsection{Proof of the Theorem}
Now, we present the proof of Theorem \ref{theorem4}.
	From Fig. \ref{fig_utility_linear}, it is obvious $h(\cdot; \theta_i)$ is concave. Consequently, Proposition~\ref{prop_equilibrium_equivalence} holds. Now, we investigate two cases. 
	
	Case (i) $\phi_{\rm max} < \frac{C}{n}$. In this case, $\phi_i \leq \phi_{\rm max}$ implies $\phi_{i} < \frac{C}{n}$ for $i \in V$, and therefore, $\sum_{i=1}^{n}\phi_{i} < C$. Consequently, Lemma \ref{lemma8} implies $\lambda^\ast = 0$. Since $\lambda^\dag >0$, one obtains $\lambda^\ast < \lambda^\dag$. Therefore, $\lambda^\ast$ is socially resilient. 
	
	Case (ii) $\phi_{\rm max} \geq \frac{C}{n}$ and $\beta_{\rm max} \leq \lambda^\dag$. 
	From Lemma \ref{lemma6}, we know that $\lambda^\ast \geq0$. If $\lambda^\ast =0$, then it is socially resilient. Therefore, we only study the case under which $\lambda^\ast >0$.
	
	According to \eqref{eq_x_linear}, if $\lambda^\ast > \beta_i$ for all $i \in V$, then $x_i^\ast = 0$ for all agents. Therefore, $\sum_{i=1}^{n}x_i^\ast = 0$, which contradicts the balancing equality $\sum_{i=1}^n x_i^\ast = C$ in \eqref{load_demand_supply_constraints} (note that $C>0$). Consequently, we obtain $\lambda^\ast \leq \beta_i$ for at least one agent. Since $\beta_i \leq \beta_{max}$ and $\beta_{max} \leq \lambda^\dag$, we yield $\lambda^\ast \leq \lambda^\dag$. 
	
	Considering  (i) and  (ii), it follows that as long as  $(\beta_{\rm max}, \phi_{\rm max})$ is constrained in the set $\mathscr{S}_{\ast}$ in \eqref{eq_set_3}, $\lambda^\ast$ will be socially resilient.

\subsection{Lemmas for MTES-ST}\label{Appendix_D}

\begin{lemma}\label{lemma5}
Consider the MTES-ST with the quadratic utility function defined in Assumption 1. The optimal load allocation $x_i^\ast$, which is an optimal solution of the optimization problem \eqref{opt_LTD_1}, is obtained as \eqref{eq_x}.
\end{lemma}
\begin{proof}
	Expanding the optimization problem in \eqref{opt_LTD_4} results in
	\begin{equation}\label{opt_LTD_2}
		\begin{aligned}
			\max_{{x}_i, s_i} \quad &  -\frac{1}{2}b_ix_i^2+ (m_i b_i - \lambda^\ast)x_i +\lambda^\ast a_i  - \lambda^\ast s_i \\
			{\rm s.t.}	\quad & x_i, s_i \in \mathbb{R}^{\geq0}.
		\end{aligned}
	\end{equation}
	Recall that for MTES-ST we have $\lambda^\ast \geq 0$ \cite{CDC}. Therefore, we can consider two cases.
		
	Case (i) $\lambda^ \ast =0$. In this case, \eqref{opt_LTD_2} yields $x_i^\ast = m_i$.
		
	Case (ii) $\lambda^ \ast >0$. According to Proposition~\ref{prop-equivalence}, both MTES-ST and MTES have the same load allocation decisions. Therefore, based on Lemma \ref{lemma4} we yield \eqref{eq_x} is valid when $\lambda^\ast >0$. 
		
	Finally, considering (i) and (ii), it follows that \eqref{eq_x} holds all the time for MTES-ST.
\end{proof}

\begin{lemma}\label{lemma2}
 Consider the MTES-ST with the quadratic utility function defined in Assumption 1. If $\sum_{i=1}^{n}m_i \leq C$,  then $\lambda^\ast =0$. Conversely, if $\sum_{i=1}^{n}m_i > C$ then $\lambda^\ast >0$.
\end{lemma}
\begin{proof}
As mentioned before, for MTES-ST we have $\lambda^\ast \geq 0$ \cite{CDC}. We investigate two cases.

(i) Consider $\sum_{i=1}^{n}m_i \leq C$. By contradiction, suppose $\lambda^\ast >0$. According to \eqref{eq_x}, $\lambda^\ast>0$ yields $x_i^\ast < m_i$ for $i \in V$, and therefore, $\sum_{i=1}^{n}x_i^\ast < \sum_{i=1}^{n}m_i$. Since $\sum_{i=1}^{n}m_i \leq C$, we yield $\sum_{i=1}^{n} x_i^\ast < C$, 
which contradicts the equality in \eqref{eq23}. Therefore, it follows that $\lambda^\ast =0$.

(ii) Consider $\sum_{i=1}^{n}m_i > C$. The inequality constraint $x_i + e_i \leq a_i$ in \eqref{opt_LTD_1} yields
\begin{equation}\label{eq18}
	\sum_{i=1}^{n}x_i^\ast \leq C.
\end{equation}
By contradiction, suppose $\lambda^\ast=0$. According to \eqref{eq_x}, we obtain $x_i^\ast = m_i$, and therefore, $\sum_{i=1}^{n}x_i^\ast = \sum_{i=1}^{n}m_i$. Since $\sum_{i=1}^{n}m_i > C$, we yield $\sum_{i=1}^{n} x_i^\ast > C$, which contradicts the inequality in \eqref{eq18}. Therefore, it follows that $\lambda^\ast >0$. 
\end{proof}

\begin{lemma}\label{lemma11}
Consider the MTES-ST with the piece-wise linear utility function defined in Assumption 2. The optimal load allocation $x_i^\ast$, which is an optimal solution of the optimization problem \eqref{opt_LTD_1}, satisfies \eqref{eq_x_linear}.
\end{lemma}

\begin{proof}
	The inequality constraint $x_i + e_i \leq a_i$ in the optimization problem \eqref{opt_LTD_1} can be written as 
	\begin{equation}\label{eq16}
		x_i + e_i + s_i = a_i,
	\end{equation}
	where $s_i\geq 0$ is the slack variable. Substituting $e_i = a_i - x_i - s_i$ into \eqref{opt_LTD_1} yields an equivalent form for the optimization problem as
	\begin{equation}\label{opt_LTD_6}
		\begin{aligned}
			\max_{{x}_i, s_i} \quad &  \min \{\beta_i x_i, \phi_i \beta_i\}+ \lambda^\ast (a_i  -  x_i -  s_i) \\
			{\rm s.t.}	\quad & x_i, s_i \in \mathbb{R}^{\geq0}.
		\end{aligned}
	\end{equation}
	Recall that for MTES-ST, we have $\lambda^\ast \geq 0$ \cite{CDC}.
	We investigate two cases.
  
    Case (i) $\lambda^\ast = 0$. Similar to the proof of Lemma \ref{lemma10}, case (i), the objective function in \eqref{opt_LTD_6} is strictly increasing in the interval $x_i \in [0, \phi_i]$, while constant in the interval $x_i \in [\phi_i, \infty)$. Therefore, the optimal solution is achieved at $x_i^\ast \geq \phi_i$.
	
	Case (ii) $\lambda^\ast > 0$. It is proved in Proposition~\ref{prop-equivalence} that for $\lambda^\ast >0$, both MTES-ST and MTES yield the same results. Therefore, this part of the proof is the same as cases (ii), (iii), and (iv) in the proof of Lemma \ref{lemma10}.	
\end{proof}

\begin{lemma}\label{lemma9}
	Consider the MTES-ST with the piece-wise linear utility function defined in Assumption 2. If $\sum_{i=1}^{n}\phi_i < C$, then $\lambda^\ast = 0$. Conversely, if $\sum_{i=1}^{n}\phi_i > C$, then $\lambda^\ast > 0$.
\end{lemma}

\begin{proof}
	It is known that $\lambda^\ast \geq 0$ \cite{CDC}. We investigate two cases.
	
	(i) Consider $\sum_{i=1}^{n}\phi_i < C$. By contradiction, if $\lambda^\ast >0$, then the objective function in \eqref{opt_LTD_6} is strictly decreasing with respect to $s_i$. Consequently, in order for the objective function to be maximized,  $s_i$ must be minimized, i.e.,  $s_i^\ast =0$. Following from \eqref{eq16}, we obtain $x_i^\ast + e_i^\ast = a_i$. Then, considering $\sum_{i=1}^{n} e_i^\ast = 0$ in \eqref{trading_demand_supply_constraints}, we yield
	\begin{equation}\label{eq14}
		\sum_{i=1}^{n} x_i^\ast = C,
	\end{equation}
	where $C= \sum_{i=1}^{n} a_i$. According to \eqref{eq_x_linear}, when $\lambda^\ast>0$ then $x_i^\ast \leq \phi_i$. Therefore, $\sum_{i=1}^{n} x_i^\ast \leq \sum_{i=1}^{n}\phi_i$. Since $\sum_{i=1}^{n}\phi_i < C$, we yield $\sum_{i=1}^{n} x_i^\ast < C$, 
	which contradicts the equality in \eqref{eq14}. Therefore, it follows that $\lambda^\ast =0$.
	
	(ii) Consider $\sum_{i=1}^{n}\phi_i > C$. The inequality constraint $x_i + e_i \leq a_i$ in \eqref{opt_LTD_1} yields
	\begin{equation}\label{eq17}
	 \sum_{i=1}^{n}x_i^\ast \leq C.
	\end{equation}
	By contradiction, suppose $\lambda^\ast=0$. According to \eqref{eq_x_linear}, we obtain $x_i^\ast \geq \phi_i$, and therefore, $\sum_{i=1}^{n}x_i^\ast \geq \sum_{i=1}^{n}\phi_i$. Since $\sum_{i=1}^{n}\phi_i > C$, we yield $\sum_{i=1}^{n} x_i^\ast > C$, which contradicts the inequality in \eqref{eq17}. Therefore, it follows that $\lambda^\ast >0$. 
\end{proof}

\end{document}